\documentclass{article}

\usepackage[utf8]{inputenc}

\usepackage{amssymb}
\usepackage{amsmath}
\usepackage{amsthm}

\usepackage[affil-it]{authblk}

\usepackage{pgf,tikz}
\usetikzlibrary{arrows}

\numberwithin{equation}{section}
\theoremstyle{plain}
\newtheorem{thm}{Theorem}[section]

\newtheorem{prop}[thm]{Proposition}
\newtheorem{coroll}[thm]{Corollary}

\theoremstyle{definition}
\newtheorem{definition}[thm]{Definition}

\newtheorem{remark}[thm]{Remark}
\newtheorem*{notation}{Notation}

\usepackage{xspace} 

\newcommand{\bigslant}[2]{{\raisebox{.2em}{$#1$}\left/\raisebox{-.2em}{$#2$}\right.}}

\def\Pic{{\rm Pic}}

\def\per{{\rm per}}

\begin{document}

\title{Rotor-routing orbits in directed graphs and the Picard group}

\author{Lilla T\'othm\'er\'esz%
\thanks{tmlilla@cs.elte.hu}}

\date{}

\affil{Department of Computer Science, E\"otv\"os Lor\'and University, P\'azm\'any P\'eter s\'et\'any 1/C, Budapest H-1117, Hungary}

\maketitle

\begin{abstract}

In \cite{Holroyd08}, Holroyd, Levine, M\'esz\'aros, 
Peres, Propp and Wilson characterize recurrent chip-and-rotor configurations for strongly connected digraphs. 
However, the number of steps needed to recur, and the 
number of orbits is left open for general digraphs.
Recently, these questions were ans\-wered by Pham \cite{pham_rotor}, using linear algebraic methods.
We give new, purely combinatorial proofs for these formulas.
We also relate rotor-router orbits to the chip-firing game: The number of recurrent rotor-router unicycle-orbits equals to the order of the Picard group of the graph, defined in the sense of \cite{DirRiemRoch}. Moreover, during a period of the rotor-router process, the same chip-moves happen, as during firing the period vector in the chip-firing game.

\end{abstract}

\section{Introduction} \label{sec::def}

Rotor routing is a deterministic process, that induces a walk of a chip on a 
directed graph.
It was introduced several times, see \cite{Priez96,Rabani98,twoballs}. 
An important feature of rotor-routing is its close relationship with chip-firing, and the Picard group of the graph \cite{Holroyd08}. Rotor routing can be thought of as a relaxed version of chip-firing.

Throughout this paper, let $D=(V,E)$ denote a strongly connected directed graph, and for each vertex $v$, fix a cyclic ordering of the outgoing edges from $v$. For an edge $e=(v,w)$, denote by $e^+$ the edge following $e$ in the cyclic order at $v$.
We denote the set of out-neighbors (in-neighbors) of a vertex $v$ by $\Gamma^+(v)$ ($\Gamma^-(v)$), the out-degree (in-degree) of a vertex $v$ by $d^+(v)$ ($d^-(v)$).

\begin{definition}[Rotor configuration]
 A rotor configuration on $D$ is a function $\varrho$ that assigns to each non-sink vertex $v$ an out-edge with tail $v$. We call $\varrho(v)$ the rotor at $v$.
\end{definition}

\begin{definition}[Chip-and-rotor configuration]
  A chip-and-rotor configuration is a pair $(w,\varrho)$, where $\varrho$ is a rotor configuration, and $w\in V(D)$. We imagine that a chip is placed on vertex $w$.
\end{definition}

The rotor router operation is a map that sends a chip-and-rotor configuration $(w, \varrho)$ to the chip-and-rotor configuration $(w^+, \varrho^+)$, where
$\varrho^+$ is a rotor configuration with
$$
\varrho^+(v) = \left\{\begin{array}{cl} \varrho(v) & \text{for }v\neq w,  \\
         \varrho(v)^+ & \text{for } v=w,
      \end{array} \right.
$$
and $w^+$ is the head of $\varrho^+(w)$.


A natural question to ask about the rotor-router operation is whether starting from a given chip-and-rotor configuration, and iterating the rotor-router operation, we ever arrive back to the initial configuration.

\begin{definition}
We call a chip-and-rotor configuration $(w,\varrho)$ \emph{recurrent}, if starting from $(w, \varrho)$, the rotor-router process eventually leads back to $(w, \varrho)$, and we call it \emph{transient} otherwise.
\end{definition}

Holroyd, Levine, M\'esz\'aros, Peres, Propp and Wilson \cite{Holroyd08} gave a characterization for recurrent chip-and-rotor configurations in strongly connected digraphs. To state their result, we need a definition.

\begin{definition}[unicycle \cite{Holroyd08}]
A unicycle is a chip-and-rotor configuration $(w, \varrho)$, where the set of edges $\{\varrho(v): v\in V(D)\}$ contains a unique directed cycle, and $w$ lies on this cycle.
\end{definition}

\begin{thm}[{\cite[Theorem 3.8, Lemma 3.4]{Holroyd08}}]
If $D$ is strongly connected, then the recurrent chip-and-rotor configurations
are exactly the unicycles. Moreover, the rotor-router operation is a permutation on the set of unicycles.
\end{thm}

Recently, Pham gave formulas for the number and length of these unicycle-orbits, using linear algebraic methods. In this paper, we give new, purely combinatorial proofs for these formulas, and explore the relationship of rotor-router unicycle-orbits with the chip-firing game and the Picard group.
Let us state some more definitions.




Let us identify $V(D)$ with $\{1,\dots, |V(D)|\}$.

\begin{definition}
The Laplacian matrix of a digraph $D$ is the
following matrix $L_D$:
$$L_D(i,j) = \left\{\begin{array}{cl} -d^+(i) & \text{if } i=j, \\
        d(j, i) & \text{if } i\neq j,
      \end{array} \right.$$
where $d(i, j)$ denotes the multiplicity of edges pointing from $i$ to $j$.
\end{definition}


\begin{definition}[period vector \cite{BL92}]
  For a strongly connected digraph $D$, let $\per_D$ be the unique vector $x\in \mathbb{Z}^{|V|}$ such that $L_Dx=0$, and the entries of $x$ have no non-trivial common divisor. This vector is unique because the Laplacian matrix of a strongly connected digraph has a one-dimensional kernel.
\end{definition}

We identify vectors in $\mathbb{Z}^{|V|}$ with integer valued functions from $V(D)$. According to this, we  write $\per_D(i)$ for the $i$-th coordinate of $\per_D$. 

\begin{definition}
  For a digraph $D$, and vertex $w\in V(D)$ a spanning in-arbores\-cence of $D$ rooted at $w$ is a subdigraph $D'$ such that $V(D')=V(D)$, $D'$ is acyclic, and for each vertex $v\in V(D)-w$ there is a unique directed path from $v$ to $w$ (i.e. the underlying undirected graph of $D'$ is a tree and each edge is oriented towards $w$).
\end{definition}

\section{Results}

Now we can state the theorem about the length of the period of the rotor-router process starting from a given unicycle. The statement of the theorem is equivalent to the first part of \cite[Theorem 1]{pham_rotor}. Also, it is a generalization of Lemma 4.9 of \cite{Holroyd08}, which only considers Eulerian digraphs.

\begin{thm} \label{thm::perhossz}
  Let $(w, \varrho)$ be a unicycle on the strongly connected digraph $D$. If we iterate the rotor-router operation starting from $(w, \varrho)$ until we arrive back to $(w, \varrho)$, then each vertex $v$ is reached $\per_D(v) d^+(v)$ times by the chip. Therefore, the rotor at vertex $v$ makes $\per_D(v)$ full turns, each edge $(u,v)$ is traversed $\per_D(u)$ times, and the length of the period is $\sum_{v\in V(G)}\per_D(v) d^+(v)$.
\end{thm}

\begin{coroll}
In a strongly connected digraph, each rotor-router unicycle-orbit has the same size.
\end{coroll}

\begin{proof}[Proof of Theorem \ref{thm::perhossz}]
  Since at the end of the process, we arrive back to the initial chip-and-rotor 
  state, each rotor makes some full turns. Say, the rotor at vertex $v$ makes $x(v)\in \mathbb{N}$ 
  full turns. Moreover, since the chip also returns to its initial place, each vertex receives the 
  chip the same number of times as it forwards it.
  Each vertex $v$ receives the chip $\sum_{u\in \Gamma^-(v)}x(u)$ times, and forwards it $x(v)d^+(v)$ times.
  Hence for each vertex $v$, the equation
  $$\sum_{u\in \Gamma^-(v)}x(u)=x(v)d^+(v)$$
  holds. Thinking of $x$ as a vector, we can rewrite this in the form $L_Dx=0$, 
  where $L_D$ is the Laplacian matrix of $D$. Thus $x$ is a multiple of $\per_D$.

  We need to show that $x$ equals to $\per_D$.
  At any moment, call a vertex $v$ \emph{good}, if $v$ has forwarded the chip at most $\per_G(v)\cdot d^+(v)$ times until that moment.
  Suppose $x\neq \per_D$.
  Then there is a moment, when some vertex $v$ forwards the chip for the $\per_D(v) d^+(v)+1$-th time, while all the other vertices are still good.
  This means that vertex $v$ has received the chip at most $\sum_{u\in \Gamma^-(v)}\per_D(u)=\per_D(v) d^+(v)$ times. But then it can forward it $\per_D(v) d^+(v)+1$ times only if $v=w$. 

  Now we see, that the first vertex $v$ to forward the chip $\per_G(v)\cdot d^+(v) +1$ times can only be $w$. It is enough to show, that when $w$ receives the chip for the $\per_G(w) d^+(w)$-th time, each other vertex $v$ has already forwarded the chip $\per_G(v)\cdot d^+(v)$ times, because in this case, each vertex $v$ forwarded the chip exactly $\per_G(v)\cdot d^+(v)$ times. Hence each rotor made some full turns, and we assumed that $w$ has just received the chip, thus, we are in state $(w, \varrho)$.

  Since $\varrho$ is a unicycle, by definition, $w$ lies on the unique cycle in $\{\varrho(v) : v\in V(D)\}$. Therefore, the edges $\{\varrho(v) : v\in V(D)\setminus w\}$ form a spanning in-arborescence $T$ with root $w$.

  We claim, that if at some moment all the vertices are good, and a vertex $v$ has received the chip $\per_D(v) d^+(v)$ times, then each in-neighbor $u$ of $v$ in $T$ has received and forwarded the chip $\per_D(u) d^+(u)$ times.

  This is indeed so, because $v$ received the chip at most $\sum_{u\in\Gamma^-(v)}\per_D(u)=\per_D(v) d^+(v)$ times, so to have equality, $v$ must have received the chip $\per_D(u)$ times from each in-neighbor. But for those in-neighbors $u$, where $uv\in E(T)$, the chip is forwarded towards $v$ for the $d^+(u)$-th, $2d^+(u)$-th, $\dots$ times, so from these vertices, the chip must have been forwarded $\per_D(u) d^+(u)$ times. As $w$ is a sink in $T$, $u\neq w$, so then $u$ also received the chip $\per_D(u) d^+(u)$ times.

Now take the moment, when $w$ receives the chip for the $\per_D(w) d^+(w)$ -th time.
We argued, that in this moment, all the vertices are good. Since $T$ is a spanning in-arborescence with root $w$, by iterating the above argument, we get that each vertex $v$ has forwarded the chip $\per_D(v) d^+(v)$ times.

\end{proof}

 It is well-known, that rotor-routing has a strong connection to the chip-firing game. For more information on this connection, and for an introduction to chip-firing, see \cite{Holroyd08}.
 Theorem \ref{thm::perhossz} further illustrates this connection:
 According to the theorem, in a recurrent rotor-router orbit, while the process returns to the initial unicycle, exactly the same chip-moves happen, as in chip-firing during the firing of the period vector.

 Now we give a formula for the number of rotor-router unicycle-orbits. This formula is stated also in \cite[Theorem 1]{pham_rotor}.

\begin{notation}
Let us denote by $\mathcal{T}(D,w)$ the number of spanning in-arborescences of a digraph $D$ rooted at vertex $w$.
\end{notation}

\begin{thm} \label{thm::orbitszam}
The number of rotor-router unicycle-orbits of a strongly connected digraph $D$ is $\frac{\mathcal{T}(D,w)}{\per_D(w)}$ for an arbitrary choice of $w\in V(D)$.
\end{thm}
\begin{proof}
 Take a spanning in-arborescence rooted at $w$. If we add any out-edge from $w$, and put a chip at $w$, we get a unicycle. This way we can get $d^+(w)$ unicycles from each spanning in-arborescence rooted at $w$, and for each arborescence and each choice of out-edge, these are different. Moreover, this way we get each unicycle where the chip is at $w$. This means, that there are exactly $\mathcal{T}(D,w)\cdot d^+(w)$ unicycles where the chip is at $w$.

 From Theorem \ref{thm::perhossz}, we know, that in any unicycle orbit, there are $\per_D(w)\cdot d^+(w)$ unicycles where the chip is at $w$. This means that from the $\mathcal{T}(D,w)\cdot d^+(w)$ unicycles constructed above, each orbit contains $\per_D(w)\cdot d^+(w)$. Since each such unicycle is contained in an orbit, this means that the number of orbits is $\frac{\mathcal{T}(D,w)}{\per_D(w)}$.
\end{proof}

\begin{remark}
Since the number of orbits is independent of the choice of $w$,  the value $\frac{\mathcal{T}(D,w)}{\per_D(w)}$ is also independent of $w$. From the primitivity of the period vector, it follows, that this value equals the greatest common divisor of $\{\mathcal{T}(D,v) :v\in V(D)\}$. Farrell and Levine call this value the \emph{Pham index} \cite{coeuler}.
\end{remark}

Theorem \ref{thm::orbitszam} also has a nice connection to chip-firing.
Let us denote by $\mathbb{Z}^n_0$ the subgroup of $\mathbb{Z}^n$ orthogonal to $\mathbf{1}$ (the all ones vector). 
The elements of $\mathbb{Z}^n_0$ are called \emph{divisors of degree zero}.

We say that two divisors $x$ and $y$ of degree zero are equivalent, if there exists a vector $z\in \mathbb{Z}^n$ such that $x=y+L_Dz$.

The Picard 
group of a digraph $D$ is the factor of $\mathbb{Z}^n_0$ by this equivalence relation, i.e
\begin{equation*}
\Pic^0(D)=\bigslant{\mathbb{Z}^n_0}{Im(L_D)}
\end{equation*}
where $Im(L_D)=\{L_Dv : v\in \mathbb{Z}^n\}$. Note that since $\mathbf{1}^\top L_D= \mathbf{0}$, indeed $Im(L_D)$ is a subgroup of $\mathbb{Z}^n_0$.

\begin{remark}
In the literature, one commonly talks about the \emph{sandpile group}, which is slightly different from the group above. The sandpile group is defined with respect to a base vertex, and for different base vertices, one obtains different groups. The sandpile group of a strongly connected digraph $D$ with base vertex $w$ is isomorphic to the direct product of $\Pic^0(D)$ with the cyclic group of order $\per_D(w)$ \cite{coeuler}. For undirected graphs, the sandpile group with respect to any base vertex is isomorphic to the Picard group. We have chosen to use the name Picard group to avoid confusion.
\end{remark}

\begin{prop}
The order of $\Pic^0(D)$ equals the number of rotor-router unicycle-orbits.
\end{prop}
The statement of this proposition implicilty appears also in \cite{coeuler}.
\begin{proof}
There are certain special elements in $\mathbb{Z}^n_0$ that are called $w$-reduced:
\begin{definition}[\cite{DirRiemRoch}]
For a digraph $D$, and vertex $w\in V(D)$, an element $x\in \mathbb{Z}^n_0$ is called $w$-reduced, if \\
(i) for all $v\in V(D) - w$, $x(v)\geq 0$,\\
(ii) for every $\mathbf{0}\neq f\in \mathbb{Z}^n$ with $\mathbf{0}\leq f\leq \per_D$ (coordinatewise), there exists $v\in V(D)- w$ such that $(x+L_Df)(v)<0$.
\end{definition}

Lemma 3.8 from \cite{DirRiemRoch} states, that for an arbitrary choice of $w\in V(D)$, in each equivalence class of the Picard group, there are exactly $\per_D(w)$ $w$-reduced elements.

From \cite{Holroyd08}, the spanning in-arborescences rooted at $w$ are in bijection with the $w$-reduced elements.

These two results together imply, that the Picard group has $\frac{\mathcal{T}(D,w)}{\per_D(w)}$ equivalence classes, that is, the order of the Picard group equals the number of rotor-router unicycle-orbits.
\end{proof}

\section*{Acknowledgement}
Research was supported by the MTA-ELTE Egerv\'ary Research Group and by the
Hungarian Scientific Research Fund - OTKA, K109240.

\bibliographystyle{abbrv}
\bibliography{rotor_r}

\end{document}